\documentclass[10pt]{llncs}
\usepackage{graphicx}
\usepackage{algorithm,algorithmic}
\usepackage[figuresright]{rotating}
\usepackage{amsmath,amssymb,amsfonts}
\usepackage{bm}
\usepackage{inputenc}
\usepackage{amsfonts}
\begin{document}
\title{Semi-online Scheduling with Lookahead \thanks{Supported by Veer Surendra Sai University of Technology, Burla, Odisha, 768018 INDIA.}}

%
%\titlerunning{Abbreviated paper title}
% If the paper title is too long for the running head, you can set
% an abbreviated paper title here
%
\author{Debasis Dwibedy\and
Rakesh Mohanty}
%
%\authorrunning{F. Author et al.}
% First names are abbreviated in the running head.
% If there are more than two authors, 'et al.' is used.
%
\institute{Veer Surendra Sai University of Technology, Burla, Odisha 768018, INDIA 
\email{(debasis.dwibedy, rakesh.iitmphd)}@gmail.com}
%\url{http://www.springer.com/gp/computer-science/lncs} \and
%ABC Institute, Rupert-Karls-University Heidelberg, Heidelberg, Germany\\
%\email{\{abc,lncs\}@uni-heidelberg.de}}
%
\maketitle              % typeset the header of the contribution
\begin{abstract}
The knowledge of future partial information in the form of a lookahead to design efficient online algorithms is a theoretically-efficient and realistic approach to solving computational problems. Design and analysis of semi-online algorithms with extra-piece-of-information (EPI) as a new input parameter has gained the attention of the theoretical computer science community in the last couple of decades. Though competitive analysis is a pessimistic worst-case performance measure to analyze online algorithms, it has immense theoretical value in developing the foundation and advancing the state-of-the-art contributions in online and semi-online scheduling. 
In this paper, we study and explore the impact of lookahead as an EPI in the context of online scheduling in identical machine frameworks. We introduce a $k$-lookahead model and design improved competitive semi-online algorithms. For a $2$-identical machine setting, we prove a lower bound of $\frac{4}{3}$ and design an optimal algorithm with a matching upper bound of $\frac{4}{3}$ on the competitive ratio. For a $3$-identical machine setting, we show a lower bound of $\frac{15}{11}$ and design a $\frac{16}{11}$-competitive improved semi-online algorithm.     
\end{abstract}
\keywords{Identical Machines, Lookahead, Makespan, Online, Semi-online, \hspace*{2.0cm}Scheduling.}
\section{Introduction} 
The $m$-machine scheduling is a well-studied NP-Complete problem with immense practical significance in diversified areas of computing, such as multiprocessor scheduling, network routing, transportation, task delegation, project management, and manufacturing, to name a few \cite{Pinedo:16}, \cite{Kress:18}. In this problem, we have to schedule the processing of a list of $n$ jobs on a set of $m$ machines with well-defined constraints and objectives, where $n > m$. Several practically-significant variants of the $m$-machine scheduling problem have been defined and studied in the literature based on machine models, availability of jobs, constraints, and objective functions. Offline and online scheduling are well-known variants of the $m$-machine scheduling problem.    \\
\textbf{Offline and Online Scheduling.} In offline scheduling, the whole list of jobs is available at the outset, while in online scheduling, jobs are given one by one in order \cite{Graham:66}, \cite{Dwibedy:22a}. Each received one must be scheduled immediately and irrevocably without knowledge of successive jobs with an objective to minimize the completion time of the job schedule, i.e., \textit{makespan}.\\
\textbf{Lookahead Model and its Significance.} Lookahead is an impactful realistic concept that helps improve the performance of online scheduling algorithms. When we make a scheduling decision without the knowledge of future inputs, foreseeing a part of the future can help us to make better decisions with additional information. Lookahead can be considered as an \textit{extra piece of information (EPI)} and constitutes a practically-significant model for defining semi-online scheduling problems.\\
\textbf{Semi-online Scheduling.} In practice, neither all jobs are available in advance as in the offline case, nor do they appear in an online fashion. The occurrence of jobs follows an intermediate framework known as semi-online \cite{Liu:96},\cite{Kellerer:97}. Semi-online scheduling is a variant of online scheduling, where in addition to the current job, some EPI about the future ones is given. Semi-online scheduling has practical significance in resource allocation \cite{Liu:96}, network bandwidth utilization \cite{Epstein:05} packet routing \cite{Bohm:16}, distributed data management \cite{Azar:01}, client request processing \cite{Feder:04}, manufacturing and production planning \cite{Yong:01}. In online settings, when some information about future inputs is given before allocating the resources, the semi-online models and algorithms are of considerable interest.\\
\textbf{Competitive Analysis.} The competitive analysis provides a mathematical framework to measure the performance of online and semi-online scheduling algorithms \cite{Tarjan:85}. Let $\sigma$ be a job sequence, $C_{ALG}(\sigma)$ be the makespan obtained by the semi-online algorithm $ALG$ for $\sigma$, and $C_{OPT}(\sigma)$ be the optimum makespan for $\sigma$. The algorithm $ALG$ is $c$-competitive, if there exists a positive number $c$ and a non-negative constant $b$ such that $C_{ALG}(\sigma)\leq c\cdot C_{OPT}(\sigma)+b$ for all $\sigma$. Here, $c$ is the \textit{upper bound (UB)} on the \textit{competitive ratio (CR)} and $c\geq 1$. The smaller value of $c$ indicates a better performance of $ALG$. The \textit{lower bound (LB)} on the CR for a semi-online scheduling problem specifies an instance of the problem for which no semi-online algorithm obtains a better CR than the LB. If $g$ is the $LB$ of a semi-online scheduling problem, then there must be an instance $\sigma$ such that no semi-online algorithm achieves a CR less than $g$. The common practice is to maximize the lower bound and minimize the upper bound on the CR. A semi-online algorithm is optimal if it achieves a CR matching the LB of the problem.\\
\textbf{Related Work.} The semi-online scheduling problem has been extensively investigated in the literature \cite{Dwibedy:22}, \cite{Epstein:18}, \cite{Boyar:17}. In semi-online scheduling, the objective is to improve the best-known bounds on the CR for a specific semi-online setting or to explore practically-significant new EPI to define innovative semi-online models. Here, we highlight only the most relevant semi-online settings and well-known competitive analysis results. Kellerer et al. \cite{Kellerer:97} introduced the buffer concept into semi-online scheduling. A buffer is a storage structure that allows an online algorithm to defer the scheduling decision on the current job by storing the incoming jobs until there is a space in the buffer. When the buffer is full, and a job arrives, an algorithm selects a job from the available ones and schedules it on a machine. Different ways of selecting and assigning jobs lead to various semi-online algorithms. A buffer capable of storing $k$ jobs allows an online algorithm to see $k+1$ jobs before making a scheduling decision. By considering a buffer of size $k (\geq 1)$ Kellerer et al. \cite{Kellerer:97} proposed an optimal algorithm $H_1$ and achieved a CR of $\frac{4}{3}$.  Zhang \cite{Zhang:97} proved a upper bound of $\frac{4}{3}$ on the CR by considering $k=1$. The recent works on deterministic non-preemptive semi-online scheduling with a buffer are \cite{Englert:08}, \cite{Lan:12}, \cite{Sun:13}, \cite{Ding:14}.\\ 
Several online problems have been investigated in the literature by considering lookahead as a parameter. In particular, the influence of lookahead on the competitiveness of online algorithms has been captured in well-known problems such as paging \cite{Albers:93}, \cite{Torng:95}, \cite{Braslauer:96}, $k$-server \cite{Ben:94}, bin packing \cite{Grove:95}, graphs \cite{Haldorsson:92}, \cite{Irani:90} dynamic location \cite{Chung:89}, and list update \cite{Albers:98}.\\
\textbf{Semi-online Scheduling with lookahead.} The semi-online scheduling with $k$-lookahead model enables an online algorithm to foresee the processing times of $k$ future jobs on receiving an incoming one, where $k\geq 1$. The model does not require an explicit structure like a buffer to store and access future jobs. The additional information provided by the lookahead model on future jobs helps an online algorithm minimize the makespan of the job schedule. On the other hand, the unbounded size of the lookahead increases the execution time of an algorithm due to the overheads of extra lookups. We can carefully reduce the time by specifying a bound on the lookahead size. \\ 
\textbf{Practical Significance.} 
In a dynamic project scheduling scenario, relevant jobs arrive on the fly. A project manager allocates resources to each incoming job without knowledge of the future ones. However, a manager often foresees and estimates the processing times of future incoming jobs based on intuition, experience, or good connections with the concerned client to meet the requirements. Nevertheless, efficient resource allocation is a non-trivial challenge in such a scenario, which motivates our study on semi-online scheduling with a lookahead.   \\
\textbf{Research Motivation.} The best-known CR for non-preemptive online scheduling with $2$ and $3$ identical machines settings is $\frac{3}{2}$ and $\frac{5}{3}$ respectively \cite{Graham:66}. Faigle et al. \cite{Faigle:89} proved that no online algorithm achieves a better CR than the best-known ones. An important research challenge is: how much improvement in the CR can be achieved if an online algorithm foresees some portion of future jobs. Several semi-online scheduling models and algorithms answer the above question by considering various new EPI. However, less attention has been paid to online scheduling with lookahead \cite{Mao:94}, \cite{Motwani:98}, \cite{Coleman:04}. According to our knowledge, a maiden work \cite{Zheng:23} has been reported in the literature about semi-online scheduling with lookahead in identical parallel machine settings for the makespan minimization objective. Here, the authors considered that an online algorithm knows the total processing time of first $k$ jobs a priori. By considering $\sum_{i=1}^{k}{p_i}\geq \alpha\delta$, the authors achieved a UB of $\frac{\alpha+2}{\alpha+1}$ on the CR, where $p_i$ is the processing time of a job $J_i$, $\alpha \geq 2$, and $\delta$ is the largest processing time. In this paper, we introduce a new realistic lookahead model and address the following non-trivial research challenge: how much lookahead is sufficient for a semi-online algorithm to beat the best-known CR in $2$ and $3$ identical machine settings? \\
\textbf{Our Contribution.} 
First, we define the $k$-lookahead model in the context of semi-online scheduling with a schematic illustration. In addition, we prove a lower bound of $\frac{4}{3}$ on the CR for a $2$-identical machine setting with $k\geq 1$. Subsequently, we design an optimal deterministic semi-online algorithm named 2-${LA}_1$ and achieve a matching UB on the CR with $k=1$. Furthermore, we highlight interesting remarks and critical observations on the $1$-lookahead in the context of semi-online scheduling on a $2$-identical machine setting. Next, we consider the $3$-identical machine setup and design a $\frac{16}{11}$-competitive deterministic semi-online algorithm named 3-${LA}_1$. We conclude by exploring several research challenges for prospective future work.  
\section{Preliminaries}
In this paper, we use the standard terminologies and notations defined in a recent survey article on semi-online scheduling \cite{Dwibedy:22}. To denote a semi-online scheduling problem setting, we follow the $3$-field ($\alpha|\beta|\gamma$) notation framework of Graham et al. \cite{Graham:79}.
\subsection{$m$-machine Scheduling Problem ($P_m\hspace*{0.1cm}|\hspace*{0.1cm} C_{max}$)}
We define the $m$-machine offline scheduling problem as follows.\\\\
\textit{Given a set $M=\{M_1, M_2, \ldots, M_m\}$ of $m (\geq 2)$ identical parallel machines and a list $J=\langle J_1, J_2, \ldots, J_{n-1}, J_n \rangle$ of $n$ independent jobs, where each job $J_i$ is characterized by its processing time $p_i$. We have to schedule $n$ jobs on $m$ machines such that the makespan $C_{max}$, or the maximum load on any machine is minimized, where load $l_j$ of a machine $M_j$ is the total sum of processing times of the jobs assigned to the machine $M_j$. We denote the problem as $P_m\hspace*{0.1cm}|\hspace*{0.1cm} C_{max}$.}\\\\
\textbf{Optimal Offline Algorithm} assigns a given list of $n$ jobs on $m$ machines such that $\max\{l_j \hspace*{0.1cm} | \hspace*{0.1cm}1\leq j\leq m\}$ is the minimal. \\
In the context of the $m$-machine scheduling problem, it is an open problem to formally define the generic behavior of the optimal offline algorithm. However, the following lower bounds of the optimal makespan $C_{OPT}$ have been used in the literature while computing the CR of online and semi-online scheduling algorithms.\\
\begin{lemma} \cite{Graham:66}.\textit{ Let $p_{max}=\max\{p_i \hspace*{0.1cm}|\hspace*{0.1cm} 1\leq i\leq n\}$. The optimal makespan $C_{OPT}$ for the problem $P_m\hspace*{0.1cm}|\hspace*{0.1cm} C_{max}$ is such that 
%$\max \{p_{max}, \frac{1}{m}\cdot \sum_{i=1}^{n}{p_i}\}$.
\[
    C_{OPT} \geq 
\begin{cases}
   p_{max}\\
   \frac{1}{m}\cdot \sum_{i=1}^{n}{p_i}\\\
\end{cases}
\]}
\end{lemma}
\textbf{Corollary 1.1.} \textit{For any instance $\sigma$ of $P_m \hspace*{0.1cm}|\hspace*{0.1cm} C_{max}$, we have $C_{OPT}\geq \max \{p_{max}, \frac{1}{m}\cdot \sum_{i=1}^{n}{p_i}\}$.}\\\\
\textbf{Corollary 1.2.} \textit{Let $l_j$ be the load of machine $M_j$ and $C_{ALG}(\sigma)$ be the makespan incurred by any algorithm $ALG$ for an instance $\sigma$. Then $C_{ALG}(\sigma)=\max\{l_j \hspace*{0.1cm} | \hspace*{0.1cm}1\leq j\leq m\}$.}\\ \\
Next, we will formally define the lookahead model in the context of semi-online scheduling by considering the definition of $P_m\hspace*{0.1cm}|\hspace*{0.1cm} C_{max}$. 
\subsection{Our Proposed Semi-online Scheduling Model with $k$-lookahead}
A semi-online algorithm receives $n$ independent jobs one by one over a list. Upon receiving a job $J_i$, the algorithm in addition to $p_i$ knows $p_{i+1}, p_{i+2}, \ldots, p_{i+k}$, where $1\leq k < n$ and $p_i >0$, $\forall i$. The received one must be scheduled non-preemptively and irrevocably  without a further clue about $p_{i+(k+1)}, p_{i+(k+2)}, \ldots, p_n$. The objective is to minimize the makespan, denoted by $C_{max}$.\\
\textbf{Illustration of the model with $1$-lookahead $({LA}_1)$.} Given $2$ identical parallel machines. Let us consider an instance $\sigma$ with a list of three jobs arriving one by one, where $\sigma=\langle J_1, J_2, J_3 \rangle$, and $p_1=1$, $p_2=1$, and $p_3=2$. A semi-online algorithm $ALG$ receives job $J_1$ and schedules $J_1$ irrevocably with knowledge of $p_1$ and $p_2$. Immediately after the scheduling of $J_1$, the job $J_2$ arrives, and $ALG$ schedules $J_2$ with knowledge of $p_2$ and $p_3$. Finally, the last job $J_3$ arrives, and $ALG$ assigns $J_3$ with prior knowledge of $p_3$. We capture the timing diagram of the scheduling of $\sigma$ in Figure \ref{fig:illustration.png}. The schedule generated by Graham's LS algorithm without lookahead is shown in Figure \ref{fig:illustration.png}(a), and the schedule generated by the optimal load balancing approach with $1$-lookahead is presented in Figure \ref{fig:illustration.png}(b). It can be observed that an online algorithm with $1$-lookahead can achieve a better makespan than its optimal online counterpart. The extra knowledge of $p_{i+1}$ while scheduling job $J_i$  substantiates minimizing the makespan.
\begin{figure}[h]
\centering
\includegraphics[scale=0.8]{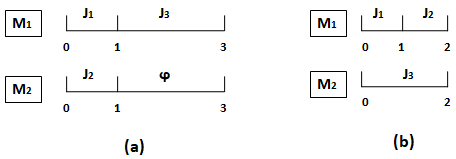}
\caption{Timing Diagram (a) Algorithm LS without Lookahead \hspace{0.1cm} (b)\hspace{0.1cm} Optimal Load Balancing with $1$-Lookahead }
\label{fig:illustration.png}  
\end{figure}
\section{Our Competitive Analysis Results on $2$ Identical Machines with $k$-Lookahead ($P_2\hspace{0.1cm}|\hspace{0.1cm}{LA}_k \hspace{0.1cm}| \hspace{0.1cm}C_{max}$)}
\subsection{Lower Bound Result}
\begin{theorem}
Let $ALG$ be a deterministic semi-online algorithm for the problem $P_2\hspace{0.1cm}|\hspace{0.1cm}{LA}_k \hspace{0.1cm}| \hspace{0.1cm}C_{max}$. Then there exists an instance $\sigma$ such that \hspace*{0.2cm}  $\frac{C_{ALG}(\sigma)}{C_{OPT}(\sigma)}\geq \frac{4}{3}$.\\
\end{theorem}
\begin{proof}
We construct an instance $\sigma = \langle J_1, J_2, \ldots, J_n \rangle$, where $p_i=x$ for $1\leq i\leq n-k-1$, $p_i=1$, for $n-k\leq i\leq n-1$, and $p_n=y$. \\\\
Given the instance $\sigma$, we compute $C_{OPT}(\sigma)$ and $C_{ALG}(\sigma)$. When the $(n-k)^{th}$ job arrives, the jobs $J_1, J_2, \ldots, J_{n-k-1}$ have been scheduled, and the processing time $p_{n-k}, p_{n-k+1}, \ldots, p_{n-1}, p_n$ are known. At this time, let the loads of machines $M_1$ and $M_2$ be $l_1$ and $l_2$. We consider the follwing two cases by assuming $l_1\geq l_2$.\\\\
\hspace*{0.4cm}\textbf{Case 1:} $l_1\geq 2l_2$. Consider $p_n=y=k$. We have $C_{ALG}(\sigma)\geq l_1$, while \\ \hspace*{1.6cm} $C_{OPT}(\sigma)=\frac{l_1+l_2+2k}{2}\leq \frac{l_1+(l_1/2)+k}{2}\leq \frac{3l_1+4k}{4}$. Implies,\\\\
\hspace*{1.6cm} $\frac{C_{ALG}(\sigma)}{C_{OPT}(\sigma)}\geq \frac{4l_1}{3l_1+4k} \rightarrow \frac{4}{3}$, for large $n$.\\\\
\hspace*{0.3cm} \textbf{Case 2:} $l_1 < 2l_2$. Consider $p_n=y=2l_1-l_2$. We have $C_{ALG}(\sigma)\geq \hspace*{1.9cm} l_2+p_n=2l_1$, while \\   
\hspace*{1.8cm} $C_{OPT}(\sigma)=\frac{l_1+l_2+k+2l_1-l_2}{2}\leq \frac{3l_1+k}{2}$. Implies,\\\\
\hspace*{1.8cm} $\frac{C_{ALG}(\sigma)}{C_{OPT}(\sigma)}\geq \frac{4l_1}{3l_1+k} \rightarrow \frac{4}{3}$, for large $n$. \hfill\(\Box\)\\
\end{proof}
%\section{Our Results on Semi-online Scheduling with ${LA}_1$}
%To be Written...\\
\textbf{Research Question:} \textit{How much lookahead is sufficient to achieve a upper bound of $\frac{4}{3}$ on the CR for $P_2\hspace{0.1cm}|\hspace{0.1cm}{LA}_k \hspace{0.1cm}| \hspace{0.1cm}C_{max}$?}\\\\
In the following section, we will address the above research question.
\subsection{An Optimal Semi-online Algorithm with $1$-lookahead : 2-${LA}_1$}
We design a deterministic semi-online algorithm named 2-${LA}_1$ for $2$ identical parallel machines setting by considering a lookahead of size $1$. We prove that algorithm 2-${LA}_1$ is optimal for the problem $P_2\hspace*{0.1cm|\hspace*{0.1cm}{LA}_1 \hspace{0.1cm}|\hspace{0.1cm}C_{max}}$ and has a upper bound of $\frac{4}{3}$ on the CR. We present algorithm 2-${LA}_1$ in \textbf{Algorithm 1}.\\
\begin{algorithm}
\caption{2-${LA}_1$}
\begin{algorithmic}
\scriptsize
\STATE Initially, $l_1=l_2=0$ \\
\STATE When a new job $J_{i}$ arrives, $p_i$ and $p_{i+1}$ are known.\\
%\STATE \hspace*{0.2cm} BEGIN\\
\STATE FOR $i=1$ to $n-1$\\
\STATE \hspace*{0.2cm} BEGIN\\
\STATE \hspace*{0.5cm} IF $(l_1+p_i)\leq \frac{2}{3}\cdot (l_1+l_2+p_i+p_{i+1})$  \\
\STATE \hspace*{0.8cm} THEN assign job $J_i$ to machine $M_1$ \\
\STATE \hspace*{0.8cm} UPDATE $l_1=l_1+p_i$\\
\STATE \hspace*{0.5cm} ELSE \\
\STATE \hspace*{0.8cm} Assign job $J_i$ to machine $M_2$ \\
\STATE \hspace*{0.8cm} UPDATE $l_2=l_2+p_i$\\
\STATE \hspace*{0.2cm} END\\
\STATE $l_{min}\leftarrow \{l_1, l_2\}$\\
\STATE Assign job $J_n$ to machine $M_j$ for which $l_j=l_{min}$, where $j=\{1, 2\}$\\
\STATE UPDATE $l_j=l_j+p_i$\\
\STATE UPDATE $l_1$, $l_2$\\
\STATE Return \hspace*{0.3cm} $C_{2-{LA}_1}=\max\{l_1, l_2\}$
\end{algorithmic}
\end{algorithm}\\
Algorithm 2-${LA}_1$ considers the EPI $p_{i+1}$ for efficient scheduling of each incoming job $J_i$. The objective is to show that $k=1$ is sufficient to achieve an improved CR over the best-known upper bound of $\frac{3}{2}$ on the CR. Algorithm 2-${LA}_1$ imbalances the current total load between machines $M_1$ and $M_2$ such that $l_1\leq 2\cdot l_2$ and the final maximum load of any machine is not more than $\frac{2}{3}\cdot \sum_{i=1}^{n}{p_i}$, or $p_{max}$.\\
In Theorem 2, we prove that algorithm 2-${LA}_1$ has a CR of at least $\frac{4}{3}$ for  $P_2\hspace*{0.1cm|\hspace*{0.1cm}{LA}_1 \hspace{0.1cm}|\hspace{0.1cm}C_{max}}$.\\
\begin{theorem}
There exists an instance $\sigma$ of the problem $P_2\hspace{0.1cm}|\hspace{0.1cm}{LA}_1 \hspace{0.1cm}| \hspace{0.1cm}C_{max}$ such that $\frac{C_{2-{LA}_1}(\sigma)}{C_{OPT}(\sigma)}\geq \frac{4}{3}$.
\end{theorem}
\begin{proof}
Consider an instance $\sigma$ with a sequence of $n$ independent jobs, where $p_i=1$, for $1\leq i\leq n-3$, $p_{n-2}=n$, $p_{n-1}=2n+3$, and $p_n=2n$.\\\\
Initially, the load $l_1=l_2=0$. In $\sigma$, jobs are given one by one in order, and each job $J_i$ from $i=1$ to $n-1$ satisfies the condition $l_1+p_i \leq \frac{2}{3}\cdot (l_1+l_2+p_i+p_{i+1})$. Therefore, algorithm 2-${LA}_1$ schedules each incoming job $J_i$ to machine $M_1$, where $1\leq i\leq n-1$.\\
We now have $l_1=4n$ and $l_2=0$. Algorithm 2-${LA}_1$ assigns the last job $J_n$ to the current least loaded machine $M_2$. Hence, $C_{2-{LA}_1}(\sigma)\geq 4n$, while $C_{OPT}=3n$. Therefore, $\frac{{C_{2-{LA}_1}(\sigma)}}{{C_{OPT}(\sigma)}}\geq \frac{4}{3}$. \hfill\(\Box\)
\end{proof}
\subsection{Upper Bound Result}
Before proving the upper bound on the CR of algorithm 2-${LA}_1$, we consider some practically-significant problem instances. We show by following lemmas that algorithm 2-${LA}_1$ is at most $\frac{4}{3}$-competitive for the considered instances of $P_2\hspace*{0.1cm|\hspace*{0.1cm}{LA}_1 \hspace{0.1cm}|\hspace{0.1cm}C_{max}}$.\\
\begin{lemma}
Let $\sigma_1$ be an instance of $P_2\hspace{0.1cm}|\hspace{0.1cm}{LA}_1 \hspace{0.1cm}| \hspace{0.1cm}C_{OPT}$, where $p_i=x$, $\forall J_i$, and $x\geq 1$.  Algorithm 2-${LA}_1$ is such that $\frac{C_{2-{LA}_1}(\sigma_1)}{C_{max}(\sigma_1)}\leq \frac{4}{3}$.
\end{lemma}
\begin{proof}
Let us consider $n$ jobs in $\sigma_1$. Given that $\sum_{i=1}^{n}{p_i}=n\cdot x$ as $p_i=x$, $\forall J_i$. Irrespective of $n$ is even or odd, algorithm 2-${LA}_1$ assigns $\lfloor \frac{2n}{3} \rfloor$ jobs to machine $M_1$ and $\lceil \frac{n}{3} \rceil$ jobs to machine $M_2$. Therefore, $C_{2-{LA}_1}(\sigma_1)\leq \frac{2n}{3}\cdot x$. If $n$ is even, algorithm $OPT$ schedules $\frac{n}{2}$ jobs to each of the machines, and incurs $C_{OPT}(\sigma_1)=\frac{n}{2}\cdot x$. Therefore, $\frac{C_{2-{LA}_1}(\sigma_1)}{C_{OPT}(\sigma_1)}\leq \frac{4}{3}$. \\\\
If $n$ is odd, algorithm $OPT$ assigns $\lceil \frac{n}{2} \rceil$ jobs to one machine and the remaining jobs to the other, incurring $C_{OPT}(\sigma_1)=\lceil \frac{n}{2} \rceil \cdot x$, while $C_{2-{LA}_1}(\sigma_1)\leq \frac{2n}{3}\cdot x$.\\ Therefore, $\frac{C_{2-{LA}_1}(\sigma_1)}{C_{OPT}(\sigma_1)}\leq \frac{4}{3}$. \hfill\(\Box\)\\
\end{proof}
\textbf{Corollary 2.1.} \textit{Let $\sigma_1$ consists of $n$ jobs, where $n=6x$, $x\geq 1$, and $p_i=1$, $\forall i$. \hspace*{2.3cm}Then $\frac{C_{2-{LA}_1}(\sigma_1)}{C_{OPT}(\sigma_1)}\leq \frac{4}{3}$}.\\\\
\textbf{Remark 1:} No algorithm can outperform Graham's LS strategy for scheduling a sequence of equal length jobs in $2$-identical machine setting. Therefore, LS algorithm is optimal for $\sigma_1$. \\\\
\begin{theorem}
Let $\sigma$ be an instance of $P_2\hspace{0.1cm}|\hspace{0.1cm}{LA}_1 \hspace{0.1cm}| \hspace{0.1cm}C_{max}$. Algorithm 2-${LA}_1$ is such that $\frac{C_{2-{LA}_1}(\sigma)}{C_{OPT}(\sigma)}\leq \frac{4}{3}$ for all $\sigma$.
\end{theorem}
\begin{proof}
We use the method of contradiction to prove the theorem. \\
Assume that there exists an instance $\sigma$ with minimum number of jobs, which contradicts the theorem, where $\sigma=\langle J_1, J_2, \ldots , J_t\rangle$. Implies, \\\\ \hspace*{4.2cm}$\frac{C_{2-{LA}_1}(\sigma)}{C_{OPT}(\sigma)} > \frac{4}{3}$. \\\\
Suppose, in the schedule generated by algorithm 2-${LA}_1$ for $\sigma$, job $J_x$ is the last job that completes its execution and $x<t$. Therefore, the instance $\sigma_1=\langle J_1, J_2, \ldots , J_x\rangle$ makes $C_{2-{LA}_1}(\sigma_1)=C_{2-{LA}_1}(\sigma)$, while $C_{OPT}(\sigma_1) < C_{OPT}(\sigma)$.\\\\ Implies, 
\hspace*{0.2cm} $\frac{C_{2-{LA}_1}(\sigma_1)}{C_{OPT}(\sigma_1)} > \frac{C_{2-{LA}_1}(\sigma)}{C_{OPT}(\sigma)} > \frac{4}{3}$. \\\\
This conveys, $\sigma_1$ is the smallest counterexample with minimum number of jobs, which contradicts our assumption on minimality of $\sigma$ with $t$ jobs. Hence, job $J_x$ does not exist and the instance $\sigma$ is the smallest counterexample in terms of number of jobs such that $\frac{C_{2-{LA}_1}(\sigma)}{C_{OPT}(\sigma)} > \frac{4}{3}$.\\ Therefore, for all other instances $\sigma^{'}$ with $t-1$ jobs, $\frac{C_{2-{LA}_1}(\sigma^{'})}{C_{OPT}(\sigma^{'})} \leq \frac{4}{3}$.\\\\
Algorithm 2-${LA}_1$ imbalances the total load on machines $M_1$ and $M_2$ such that $l_1\geq l_2$ and $l_1\leq \frac{2}{3}\cdot \sum_{i=1}^{t-1}{p_i}$. Therefore, we have $\frac{C_{2-{LA}_1}(\sigma^{'})}{C_{OPT}(\sigma^{'})} \leq \frac{4}{3}$. Clearly, the inclusion and scheduling of job $J_t$ makes $\frac{C_{2-{LA}_1} (\sigma)}{C_{OPT}(\sigma)} > \frac{4}{3}$.\\\\
Prior to scheduling $J_t$, we have $l_{max}=\frac{2}{3}\cdot \sum_{i=1}^{t-1}{p_i}$ and $l_{min}=\frac{1}{3}\cdot \sum_{i=1}^{t-1}{p_i}$.\\\\ Algorithm 2-${LA}_1$ schedules $J_t$ on $M_j$ with $l_j=l_{min}$. Implies, \\\\
\hspace*{2.2cm} $l_{min}+p_t > \frac{2}{3}\cdot \sum_{i=1}^{t}{p_i}$\\
$\implies \hspace{0.3cm}\frac{1}{3}\cdot \sum_{i=1}^{t-1}{p_i}+p_t > \frac{2}{3}\cdot \sum_{i=1}^{t-1}{p_i}+\frac{2}{3}\cdot p_t$ \\\\
Without loss of generality, we assume that $l_2 < l_1$, implies, $l_2+p_t > l_1+\frac{2}{3}\cdot p_t$\\
$\implies \hspace{0.3cm} l_2+p_t > 2\cdot l_2+\frac{2}{3}\cdot p_t$\\
$\implies \hspace*{0.3cm} p_t > l_1+l_2$, implies, $p_t > \sum_{i=1}^{t-1}{p_i}$.\\\\ In such a case algorithm 2-${LA}_1$ schedules each incoming job $J_i$ from $i=1$ to $t-1$ on machine $M_1$ and assigns job $J_t$ to machine $M_2$ to incur $C_{2-{LA}_1}(\sigma)=p_t$, while $C_{OPT}(\sigma)=p_t$. Implies, $\frac{C_{2-{LA}_1}(\sigma)}{C_{OPT}(\sigma)}=1\leq \frac{4}{3}$, which is a contradiction on our assumption on $\sigma$.\\\\
Therefore, there does not exist a counterexample to the theorem. We can now conclude that the theorem holds for all instances. \hfill\(\Box\)
\end{proof}
\textbf{Observations}
\begin{itemize}
\item Increasing the value of lookahead from $1$ to any $k$ would not lead to a better CR than $\frac{4}{3}$.
\item The proposed $1$-lookahead model and algorithm consider the arrival and scheduling of one job at a time with knowledge of non-zero positive processing time of the current and the next job. The model can not achieve a better CR than $\frac{3}{2}$ if it considers $p_i=0$ for any $J_i$, and scheduling of $J_{i+1}$ without knowledge of $p_{i+2}$. For example, consider an instance $\sigma=\langle J_1, J_2, J_3 \rangle$ with $p_1=p_2=1$, and $p_3=y$. Let $ALG$ be a semi-online algorithm. If $ALG$ assigns $J_1$ and $J_2$ to the same machine, then consider $p_3=y=0$. Implies, $C_{ALG}(\sigma)=2$, while $C_{OPT}(\sigma)=1$.\\
If $ALG$ schedules $J_1$ and $J_2$ on different machines, then consider $p_3=y=2$. Implies, $C_{ALG}(\sigma)=3$, while $C_{OPT}(\sigma)=2$. Therefore, $\frac{C_{ALG}(\sigma)}{C_{OPT}(\sigma)}\geq \frac{3}{2}$.\\
Hence, our model has significance over the traditional lookahead models concerning performance improvement.
\end{itemize}
\section{Our Competitive Analysis Results on $3$ Identical Machines with $1$-Lookahead ($P_3\hspace{0.1cm}|\hspace{0.1cm}{LA}_1 \hspace{0.1cm}| \hspace{0.1cm}C_{max}$)}
Graham \cite{Graham:66} explored the following multiprocessing timing anomaly. An increase in the number of machines does not necessarily guarantee the minimization of makespan. Graham proved that the LS algorithm is the optimal one with CR of $\frac{5}{3}\approx 1.66$ for $3$-identical machine setting. In this section, we study the problem $P_3\hspace*{0.1cm}|\hspace*{0.1cm}{LA}_1 \hspace*{0.1cm}|\hspace*{0.1cm}C_{max}$ and achieve an improved CR of $\frac{16}{11}\approx 1.45$.
\subsection{Lower Bound Result}
\begin{theorem}
Let $ALG$ be a deterministic semi-online algorithm for the problem $P_3\hspace{0.1cm}|\hspace{0.1cm}{LA}_1 \hspace{0.1cm}| \hspace{0.1cm}C_{max}$. Then there exists an instance $\sigma$ such that $\frac{C_{ALG}(\sigma)}{C_{OPT}(\sigma)}\geq \frac{15}{11}$.
\end{theorem}
\begin{proof}
We consider the adverserial instances of the problem to prove the lower bound. Let us assume that the adversary knows the nature of $ALG$ and dynamically generates an instance $\sigma$ such that $\frac{C_{ALG}(\sigma)}{C_{OPT}(\sigma)}\geq \frac{15}{11}$. We examine through the following cases, various strategies of scheduling $\sigma$ that $ALG$ could employ. We use the notation $J_i/p_i$ to denote a job $J_i$ and its corresponding processing time $p_i$. Without loss of generality, let us consider an instance $\sigma=\{J_1/7, J_2/4, J_3/p_3, \ldots, J_n/p_n\}$, where $n\geq 3$. Intially, the loads $l_1=l_2=l_3=0$.\\\\
\textbf{Case 1:} If $ALG$ assigns $J_1$, $J_2$ and $J_3$ to three different machines  (say $M_1$, $M_2$ and $M_3$ respectively), or to the same machine. Consider $p_3=4$, $p_4=7$ and $p_5=11$. We have $C_{ALG}(\sigma)\geq 15$, while $C_{OPT}(\sigma)=11$. Therefore, $\frac{C_{ALG}(\sigma)}{C_{OPT}(\sigma)}\geq \frac{15}{11}$.\\\\
\textbf{Case 2:} If $ALG$ assigns $J_1$, $J_2$ to the same machine (say $M_1$) and $J_3$ to a different machine (say $M_2$). Consider $p_3=4$, now $l_1=11$, $l_2=4$ and $l_3=0$.\\
\hspace*{1.0cm} \textbf{Sub-case 2.1:} If $J_4$ is assigned to $M_3$. Consider $p_4=7$ and $p_5=11$.\\ 
\hspace*{1.0cm} We have $C_{ALG}(\sigma)\geq 15$, while $C_{OPT}(\sigma)=11$. Therefore, $\frac{C_{ALG}(\sigma)}{C_{OPT}(\sigma)}\geq \frac{15}{11}$.\\
\hspace*{1.0cm} \textbf{Sub-case 2.2:} If $J_4$ is assigned to $M_2$. Consider $p_4=11$ and $p_5=7$.\\
\hspace*{1.0cm} Implies, $C_{ALG}(\sigma)\geq 15$, while $C_{OPT}(\sigma)=11$.  Therefore, $\frac{C_{ALG}(\sigma)}{C_{OPT}(\sigma)}\geq \frac{15}{11}$. \\ 
\hspace*{1.0cm} \textbf{Sub-case 2.3:} If $J_4$ is assigned to $M_1$. Consider $p_4=4$ and $p_5=11$.\\
\hspace*{1.0cm} Clearly, $C_{ALG}(\sigma)\geq 15$, while $C_{OPT}(\sigma)=11$.  Therefore, $\frac{C_{ALG}(\sigma)}{C_{OPT}(\sigma)}\geq \frac{15}{11}$. \\\\
\textbf{Case 3:} If $ALG$ assigns $J_1$, $J_2$ to different machines (say $M_1$ and $M_2$ respectively).\\
\hspace*{1.0cm} \textbf{Sub-case 3(a):} Job $J_3$ is assigned to $M_1$. Consider $p_3=4$, now $l_1=11$,\\
\hspace*{1.0cm} $l_2=4$ and $l_3=0$.\\
\hspace*{1.3cm} \textbf{Sub-case 3(a).1:} If $J_4$ is assigned to $M_3$. Consider $p_4=7$ and $p_5=11$.\\
\hspace*{1.3cm} We have $C_{ALG}(\sigma)\geq 15$, while $C_{OPT}(\sigma)=11$. Therefore, $\frac{C_{ALG}(\sigma)}{C_{OPT}(\sigma)}\geq \frac{15}{11}$.\\
\hspace*{1.3cm} \textbf{Sub-case 3(a).2:} If $J_4$ is assigned to $M_1$. Consider $p_4=4$ and $p_5=11$.\\
\hspace*{1.3cm} Clearly, $C_{ALG}(\sigma)\geq 15$, while $C_{OPT}(\sigma)=11$.  Therefore, $\frac{C_{ALG}(\sigma)}{C_{OPT}(\sigma)}\geq \frac{15}{11}$.\\
\hspace*{1.3cm} \textbf{Sub-case 3(a).3:} If $J_4$ is assigned to $M_2$. Consider $p_4=11$ and $p_5=7$.\\
\hspace*{1.3cm} Implies, $C_{ALG}(\sigma)\geq 15$, while $C_{OPT}(\sigma)=11$. Therefore, $\frac{C_{ALG}(\sigma)}{C_{OPT}(\sigma)}\geq \frac{15}{11}$. \\\\
\hspace*{1.0cm} \textbf{Sub-case 3(b):} Job $J_3$ is assigned to $M_2$. Consider $p_3=4$, now $l_1=7$,\\
\hspace*{1.0cm} $l_2=8$ and $l_3=0$.\\
\hspace*{1.3cm} \textbf{Sub-case 3(b).1:} If $J_4$ is assigned to $M_2$ or $M_3$. Consider $p_4=7$ and \\
\hspace*{1.3cm} $p_5=8$. We have $C_{ALG}(\sigma)\geq 15$, while $C_{OPT}(\sigma)=11$.   \\
\hspace*{1.3cm} \textbf{Sub-case 3(b).2:} If $J_4$ is assigned to $M_1$. Consider $p_4=8$ and $p_5=7$.\\
\hspace*{1.3cm} Implies, $C_{ALG}(\sigma)\geq 15$, while $C_{OPT}(\sigma)=11$.  Therefore, $\frac{C_{ALG}(\sigma)}{C_{OPT}(\sigma)}\geq \frac{15}{11}$. \\\\
Therefore, we conclude that there exists an instance $\sigma$ of the problem  $P_3\hspace{0.1cm}|\hspace{0.1cm}{LA}_1 \hspace{0.1cm}| \hspace{0.1cm}C_{max}$ such that $\frac{C_{ALG}(\sigma)}{C_{OPT}(\sigma)}\geq \frac{15}{11}$. \hfill\(\Box\)
\end{proof}
\subsection{An Improved Semi-online Algorithm with $1$-lookahead : 3-${LA}_1$ }
We design a deterministic semi-online algorithm named 3-${LA}_1$ for $3$ identical parallel machines setting by considering a lookahead of size $1$. We prove that algorithm 3-${LA}_1$ has a CR of at most $\frac{16}{11}\approx 1.45$. We present algorithm 3-${LA}_1$ in \textbf{Algorithm 2}.
\begin{algorithm}
\caption{3-${LA}_1$}
\begin{algorithmic}
\scriptsize
\STATE Initially, $l_1=l_2=l_3=0$ \\
\STATE When a new job $J_{i}$ arrives, $p_i$ and $p_{i+1}$ are known.\\
%\STATE \hspace*{0.2cm} BEGIN\\
\STATE FOR $i=1$ to $n-1$\\
\STATE \hspace*{0.2cm} BEGIN\\
\STATE \hspace*{0.5cm} IF $(l_1+p_i)\leq \frac{16}{33}\cdot (l_1+l_2+l_3+p_i+p_{i+1})$  \\
\STATE \hspace*{0.8cm} THEN assign job $J_i$ to machine $M_1$ \\
\STATE \hspace*{0.8cm} UPDATE $l_1=l_1+p_i$\\
\STATE \hspace*{0.5cm} ELSE IF $(l_2+p_i)\leq \frac{15}{33}\cdot (l_1+l_2+l_3+p_i+p_{i+1})$  \\
\STATE \hspace*{0.8cm} Assign job $J_i$ to machine $M_2$ \\
\STATE \hspace*{0.8cm} UPDATE $l_2=l_2+p_i$\\
\STATE \hspace*{0.5cm} ELSE \\
\STATE \hspace*{0.8cm} Assign job $J_i$ to machine $M_3$ \\
\STATE \hspace*{0.2cm} END\\
\STATE $l_{min}\leftarrow \min \{l_1, l_2, l_3\}$\\
\STATE Assign job $J_n$ to machine $M_j$ for which $l_j=l_{min}$, where $j=\{1, 2, 3\}$\\
\STATE UPDATE $l_j=l_j+p_i$\\
\STATE UPDATE $l_1$, $l_2$, $l_3$\\
\STATE Return \hspace*{0.3cm} $C_{3-{LA}_1}=\max\{l_1, l_2, l_3\}$
\end{algorithmic}
\end{algorithm}\\
Algorithm 3-${LA}_1$ considers the EPI $p_{i+1}$ for efficient scheduling of each incoming job $J_i$. Our goal is to show that $k=1$ is sufficient to achieve an improved CR over the best-known upper bound of $\frac{5}{3}\approx 1.66$ on the CR. The objective of algorithm 3-${LA}_1$ is to keep an imbalance in the loads of machines $M_1$, $M_2$ and $M_3$ such that the final loads $l_1\leq \frac{16}{33}\cdot \sum_{i=1}^{n}{p_i}$, $l_2\leq \frac{15}{33}\cdot \sum_{i=1}^{n}{p_i}$ and $l_3\leq C_{OPT}$. Our goal is to show that $C_{3-{LA}_1}\leq \frac{16}{33}\cdot \sum_{i=1}^{n}{p_i}$, while $C_{OPT}\geq \frac{1}{3}\cdot \sum_{i=1}^{n}{p_i}$, or $C_{OPT}\geq p_{max}$ such that $\frac{C_{3-{LA}_1}}{C_{OPT}}\leq \frac{16}{11}\approx 1.45$ for all instances of $P_3\hspace*{0.1cm|\hspace*{0.1cm}{LA}_1 \hspace{0.1cm}|\hspace{0.1cm}C_{max}}$. 
\subsection{Upper Bound Result}
We prove the competitiveness of algorithm 3-${LA}_1$ by considering several critical cases as Lemmas to establish the upper bound on the CR.\\
\begin{theorem}
Let $\sigma$ be an instance of $P_3\hspace{0.1cm}|\hspace{0.1cm}{LA}_1 \hspace{0.1cm}| \hspace{0.1cm}C_{max}$. Algorithm 3-${LA}_1$ is such that $\frac{C_{3-{LA}_1}(\sigma)}{C_{OPT}(\sigma)}\leq \frac{15}{11}$ for all $\sigma$.
\end{theorem}
\textit{Proof.} We prove Theorem 5 by Lemma 3-6 as follows. 
\begin{lemma}
Let $T=\sum_{i=1}^{n}{p_i}$ for any instance $\sigma$ of the problem $P_3\hspace*{0.1cm|\hspace*{0.1cm}{LA}_1 \hspace{0.1cm}|\hspace{0.1cm}C_{max}}$. If the final load of machine $M_3$ is such that $l_3\leq \frac{2}{33}\cdot T$, then $\frac{C_{3-{LA}_1}(\sigma)}{C_{OPT}(\sigma)}\leq \frac{16}{11}$. 
\end{lemma}
\begin{proof}
This is the obvious case, no matter the order of arrival of jobs in $\sigma$ and their processing times, we have $C_{3-{LA}_1}(\sigma)=l_{max}=l_1\leq \frac{16}{33}\cdot T$, while $C_{OPT}(\sigma)\geq \frac{1}{3}\cdot T$. Therefore, $\frac{C_{3-{LA}_1}(\sigma)}{C_{OPT}(\sigma)}\leq \frac{16}{11}$. \hfill\(\Box\)
\end{proof}
The following lemma put some insights into the maximum value of $p_{max}$ and its impact on the performance of algorithm 3-${LA}_1$.\\
\begin{lemma}
If $p_{max} \geq \frac{16}{33}\cdot T$, then $C_{3-{LA}_1}=C_{OPT}=p_{max}$. 
\begin{proof}
If $p_{max} \geq \frac{16}{33}\cdot T$, then $\sum_{}^{}{p_i}\leq \frac{17}{33}\cdot T$ for the rest $n-1$ jobs. As we consider $n\geq 3$, w.o.l.g, let us normalize the processing times of the jobs such that $T=33$. Clearly, $p_{max}\leq \frac{31}{33}\cdot T$. Consider the sequence $\sigma=\langle J_1/16, J_2/16, J_3/1 \rangle$ with least number of jobs. We have, ${3-{LA}_1}(\sigma)=C_{OPT}(\sigma)=p_{max}$.\\ \hfill\(\Box\)
\end{proof}
\end{lemma}
\begin{lemma}
If $p_{max}=p_1$ or $p_n$, and $p_{max}\geq \sum_{i}^{}{p_i}$ for rest $n-1$ jobs, then $C_{3-{LA}_1}=C_{OPT}=p_{max}$. 
\end{lemma}
\begin{proof}
Following the proof of the previous lemma, let us consider the sequences $\sigma_2=\langle J_1/17, J_2/14, J_3/1, J_2/1 \rangle$ and $\sigma_3=\langle J_1/1, J_2/1, J_3/14, J_2/17 \rangle$, where $p_{max}=17$. For both sequences, we have, ${3-{LA}_1}(\sigma)=C_{OPT}(\sigma)=p_{max}$. \hfill\(\Box\) 
\end{proof}
\begin{lemma}
Let $\sigma$ be an instance with $n$ jobs, where $p_i=1$, $\forall J_i$, then $\frac{C_{3-{LA}_1}(\sigma)}{C_{OPT}(\sigma)}\leq \frac{16}{11}$.
\end{lemma}
\begin{proof}
We explore and consider a worst job sequence $\sigma$ to establish our claim. Let $n=33x$, where $x\geq 1$. Here, $T=n$. Consider an instance $\sigma$, where $x=1$, implies $T=33$. We have $C_{3-{LA}_1}\leq 16$, while $C_{OPT}(\sigma)=11$. Therefore, $\frac{C_{3-{LA}_1}(\sigma)}{C_{OPT}(\sigma)}\leq \frac{16}{11}$. \hfill\(\Box\)  
\end{proof}
%\begin{theorem}
%Let $\sigma$ be an instance of $P_3\hspace{0.1cm}|\hspace{0.1cm}{LA}_1 \hspace{0.1cm}| \hspace{0.1cm}C_{max}$. Algorithm 3-${LA}_1$ is such that $\frac{C_{3-{LA}_1}(\sigma)}{C_{OPT}(\sigma)}\leq \frac{15}{11}$ for all $\sigma$.
%\end{theorem}
%\begin{proof}
%\end{proof}
\section{Concluding Remarks}
In this paper, we introduced the $k$-lookahead model in non-preemptive online scheduling for makespan minimization to beat the best-known tight bounds of $\frac{3}{2}$ and $\frac{5}{3}$ on the competitive ratio in the $m$-identical machine settings, where $m=2, 3$ respectively.  For a $2$-identical machine case, we proved that increasing the lookahead size from $1$ to any $k$ would not lead to a better tight bound than $\frac{4}{3}$ on the CR. In a $3$-identical machine setting with $1$-lookahead, we obtained an improved upper bound of $\frac{16}{11}$ and a lower bound of $\frac{15}{11}$ on the competitive ratio. Our objective in this paper was to systematically disseminate the influence of lookahead as a realistic new EPI in the performance improvement of the best-known online scheduling algorithms in the considered settings. The following interesting non-trivial research challenges will further inspire future studies on semi-online scheduling with a lookahead.\\\\
\textbf{Research Challenges}
\begin{itemize}
\item Can the gap between our achieved UB of $\frac{16}{11}$ and LB of $\frac{15}{11}$ on the CR for the problem $P_3\hspace{0.1cm}|\hspace{0.1cm}{LA}_1 \hspace{0.1cm}| \hspace{0.1cm}C_{max}$ be minimized?
\item How much lookahead is sufficient to achieve a tight bound on the CR for the $3$-identical machine setup?
\item How much effective is the $k$-lookahead model for semi-online scheduling on a $m$-identical machine setting, where $m > 3$?
\item Can the $k$-lookahead model be further characterized for special job sequences based on real-world applications?
\end{itemize}

%%\begin{thebibliography}{8}
%%\bibitem{ref_article1}
%%Author, F.: Article title. Journal \textbf{2}(5), 99--110 (2016)
%%
%%\bibitem{ref_lncs1}
%%Author, F., Author, S.: Title of a proceedings paper. In: Editor,
%%F., Editor, S. (eds.) CONFERENCE 2016, LNCS, vol. 9999, pp. 1--13.
%%Springer, Heidelberg (2016). \doi{10.10007/1234567890}
%%
%%\bibitem{ref_book1}
%%Author, F., Author, S., Author, T.: Book title. 2nd edn. Publisher,
%%Location (1999)
%%
%%\bibitem{ref_proc1}
%%Author, A.-B.: Contribution title. In: 9th International Proceedings
%%on Proceedings, pp. 1--2. Publisher, Location (2010)
%%
%%\bibitem{ref_url1}
%%LNCS Homepage, \url{http://www.springer.com/lncs}. Last accessed 4
%%Oct 2017
%%\end{thebibliography}
\end{document}